\documentclass[10pt, conference, letterpaper]{IEEEtran}
\IEEEoverridecommandlockouts
\usepackage[linesnumbered, vlined, ruled, commentsnumbered]{algorithm2e}
\usepackage{mathrsfs}
\usepackage{bm}
\usepackage[usenames, dvipsnames]{xcolor}

\usepackage{amsmath}
\usepackage{amsfonts,amssymb}
\usepackage{amsthm}
\usepackage[noend]{algpseudocode}
\usepackage{booktabs}
\usepackage{footnote}
\usepackage{colortbl}
\usepackage{tabularx}
\usepackage{multirow}
\usepackage{makecell}
\usepackage{threeparttable}
\usepackage{graphicx}
\usepackage{epstopdf}
\usepackage{subfigure}
\usepackage{hhline}
\usepackage{hyperref}
\usepackage{longtable}
\usepackage{pifont}
\usepackage{ulem}
\usepackage{wasysym}
\usepackage{graphbox}

\newtheorem{theorem}{Theorem}

\newtheorem{definition}{Definition}

\definecolor{BrickRed}{RGB}{178,34,34}

\newcommand{\cmark}{\ding{51}}%
\newcommand{\xmark}{\ding{55}}%

\makeatletter  
\newif\if@restonecol  
\makeatother  

\usepackage{fancyhdr}
\usepackage[
  n,
  operators,
  advantage,
  sets,
  adversary,
  landau,
  probability,
  notions,  
  ff,
  mm,
  primitives,
  events,
  complexity,
  asymptotics,
  keys]{cryptocode}

\begin{document}
\title{EC-Chain: Cost-Effective Storage Solution for Permissionless Blockchains
\thanks{Corresponding author: Hechuan Guo (\href{mailto:hcguo@sdu.edu.cn}{hcguo@sdu.edu.cn}).}
\thanks{This study was supported by the National Key R\&D Program of China (No.2023YFB2703600), the National Natural Science Foundation of China (No. 62302266, 62232010, U23A20302), the Shandong Science Fund for Excellent Young Scholars (No.2023HWYQ-008), and the project ZR2022ZD02 supported by Shandong Provincial Natural Science Foundation.}
}


\author{
\IEEEauthorblockN{Minghui~Xu$^\dag$,
Hechuan Guo$^\ddagger$,
Ye Cheng$^\dag$,
Chunchi Liu$^\P$,
Dongxiao Yu$^\dag$,
Xiuzhen Cheng$^\dag$
}
\IEEEauthorblockA{$^\dag$ School of Computer Science and Technology, Shandong University}
\IEEEauthorblockA{$^\ddagger$ School of Information Science and Engineering, Shandong University}
\IEEEauthorblockA{$^\P$ Huawei Technologies Co., Ltd.}

}


\markboth{Journal of \LaTeX\ Class Files,~Vol.~14, No.~8, August~2015}%
{Journal of \LaTeX\ Class Files,~Vol.~14, No.~8, August~2015}

\IEEEtitleabstractindextext{%
\begin{abstract}
Permissionless blockchains face considerable challenges due to increasing storage demands, driven by the proliferation of Decentralized Applications (DApps). This paper introduces EC-Chain, a cost-effective storage solution for permissionless blockchains. EC-Chain reduces storage overheads of ledger and state data, which comprise blockchain data. For ledger data, EC-Chain refines existing erasure coding-based storage optimization techniques by incorporating batch encoding and height-based encoding. We also introduce an easy-to-implement dual-trie state management system that enhances state storage and retrieval through state expiry, mining, and creation procedures. To ensure data availability in permissionless environments, EC-Chain introduces a network maintenance scheme tailored for dynamism. Collectively, these contributions allow EC-Chain to provide an effective solution to the storage challenges faced by permissionless blockchains. Our evaluation demonstrates that EC-Chain can achieve a storage reduction of over \(90\%\) compared to native Ethereum Geth.
\end{abstract}

}

\maketitle

\IEEEdisplaynontitleabstractindextext
\IEEEpeerreviewmaketitle

\section{Introduction}
\label{sec:introduction}
Blockchain technology, esteemed for its capacity to uphold a secure and immutable distributed ledger among a network of untrusted nodes, is increasingly confronted with the issue of storage overhead. This problem is particularly pronounced with the rising prevalence of Decentralized Applications (DApps), such as blockchain games~\cite{min2019blockchain}, Decentralized Exchanges (DEXs)~\cite{lehar2021decentralized}, and Decentralized Finance (DeFi)~\cite{werner2022sok}. These applications produce substantial data volumes daily, intensifying storage demands. For example, the Ethereum network experiences an approximate daily data increase of 0.2 GB\footnote{https://www.statista.com} per node. This swift expansion in storage requirements poses a significant challenge for computers serving as blockchain storage nodes. As storage needs grow, the number of users maintaining nodes decreases, threatening the security and decentralization that are fundamental to permissionless blockchains~\cite{kim2021ethanos}.

Blockchain storage overhead arises primarily from two sources: ledger data and state data. Ledger data pertains to the immutable chain of blocks, each containing a set of transactions. State data, representing the current system state derived from past transactions, is generally organized in a tree structure to support fast verification. Notably, blockchain nodes seldom need to access the complete historical ledger. Instead, they frequently access the state data during transaction processing~\cite{dinh2018untangling}. As a typical example, the state data of Ethereum encompasses more than just account balances. It also includes smart contract code and storage, thereby enabling a broader range of functionalities. This added functionality, however, leads to increased storage requirements. To mitigate the problem of state data explosion, Ethereum implemented state pruning with the transition from version 1.11.5 to version 1.13.8\footnote{https://blog.ethereum.org/2023/09/12/geth-v1-13-0}. Despite this improvement, state data still constitutes a significant portion (33.29\%) of total storage and continues to expand rapidly. Current research efforts aim to enhance storage solutions~\cite{sforzin2022storage, pontiveros2018recycling, boneh2019batching, du2023partitionchain, qi2020reliable}, yet we still face challenges in effectively reducing blockchain storage redundancy.

\begin{figure}[!t]
\centering
\subfigure[Storage utilization]{
\label{fig:storage:usage}
\includegraphics[width=0.48\linewidth]{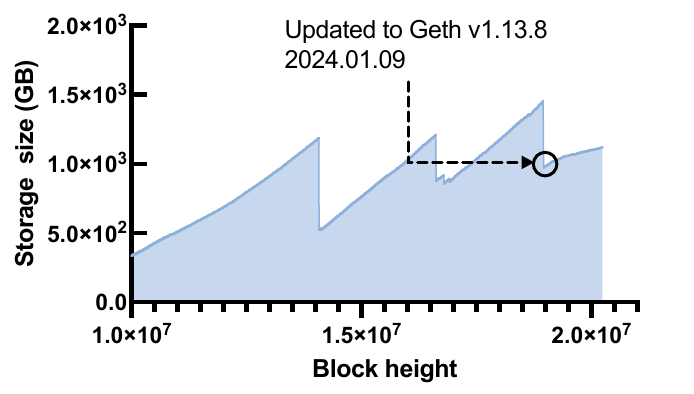}}
\subfigure[Storage utilization ratio]{
\label{fig:storage:ratio}
\includegraphics[width=0.48\linewidth]{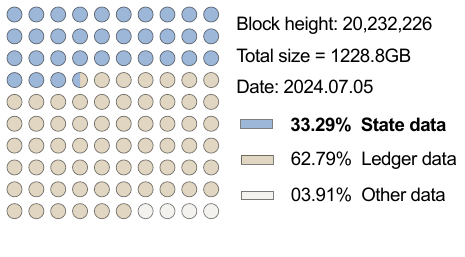}}
\caption{Storage metrics of Ethereum geth node}
\label{fig:storage}
\end{figure}

To address the challenges associated with state data, stateless blockchains have been developed, providing succinct and verifiable state proofs and mitigating the need for complex state management~\cite{boneh2019batching, srinivasan2021hyperproofs, dryja2019utreexo, bailey2021merkle, verkletree}. However, this requires off-chain state storage and frequent witness updates, which can compromise on-chain data availability and impose significant burdens on users. Moreover, stateless approaches require significant modifications to the existing state storage architecture, which is not implementation-friendly. Therefore, altering the current state-trie architecture is another promising solution, but mature implementations are still limited. This is because state data is characterized by frequent access and high availability demands, making it challenging to directly apply existing storage optimization techniques such as pruning~\cite{nakamoto2008bitcoin, sforzin2022storage}, compression~\cite{pontiveros2018recycling, kim2019scc, sforzin2022storage}, and erasure coding~\cite{qi2020reliable, du2023partitionchain, li2021lightweight}. Moreover, the dynamic nature of permissionless networks, characterized by frequent node arrivals and departures, further exacerbates data availability issues. The large scale of permissionless networks also adds extra complexity to storage optimizations.

EC-Chain specifically addresses these critical challenges in permissionless blockchains. By integrating optimized encoding schemes, and an innovative dual-trie state management system, EC-Chain reduces storage overhead for both ledger and state data. We also propose network maintenance to adapt to the dynamism of permissionless networks. Our key contributions include the following:

\begin{enumerate}
\item EC-Chain decreases storage costs for permissionless blockchains by applying erasure coding (EC) to both ledger and state data. In the context of ledger data, the system incorporates batch encoding and height-based encoding methods to improve upon current EC-based blockchain storage optimization techniques. 

\item To effectively handle the expanding volume of state data, EC-Chain implements an innovative dual-trie state management system. This system divides the original state trie into two distinct hot and cold tries to facilitate low redundancy in storage and optimize state retrieval processes. Additionally, we introduce novel techniques including state expiry, mining, and creation, which collectively contribute to the efficient management of state data.

\item EC-Chain utilizes a trie-splitting method to enhance the extraction of information from encoded segments. This technique enables the efficient retrieval of specific data without the need to recover the entire cold trie. Importantly, the integrity of the data is preserved and remains verifiable after encoding. 

\item The EC-Chain system incorporates a network maintenance scheme designed to address the dynamic nature of permissionless networks. This offers adaptable mechanisms for group upgrades and downgrades, as well as chunk updates, to improve data availability regarding frequent node arrivals and departures.
\end{enumerate}



\begin{table*}[!ht]
\centering
\begin{threeparttable}
\caption{A Comparison of Blockchain Storage Reduction Techniques}
\label{table:comparision}
\tabcolsep=0.3cm
\begin{tabular}{l|c|c|c|>{\centering}p{1.5cm}|>{\centering}p{1.5cm}|c}
\toprule[1pt]
\multirow{2}{*}{\ } & \multirow{2}{*}{\textbf{Method}} & \multirow{2}{*}{\makecell{\textbf{Permissionless} \\ \textbf{(Dynamism)}}} & \multirow{2}{*}{\textbf{Smart Contract}} & \multicolumn{2}{c|}{\centering\textbf{Storage Optimization of}} & \multirow{2}{*}{\makecell{\textbf{On-Chain Data} \\ \textbf{Availability}}} \\
 & & & & \textbf{Ledger} & \textbf{State$^\dag$} & \\
\midrule[0.5pt]
SPV \cite{nakamoto2008bitcoin} & \multirow{2}{*}{Pruning} & \cmark & \xmark & \cmark & \cmark & \LEFTcircle \\ 
MINIMIZE \cite{sforzin2022storage} & & \cmark & \xmark & \cmark & \cmark & \LEFTcircle \\ 
\midrule
Pontiveros \textit{et al.} \cite{pontiveros2018recycling} & \multirow{3}{*}{Compression} & \cmark & \cmark & \xmark & \cmark & \CIRCLE \\
SCC \cite{kim2019scc} & & \xmark & \xmark & \cmark & \xmark & \CIRCLE \\
SLACK \cite{sforzin2022storage} & & \cmark & \xmark & \cmark & \xmark & \CIRCLE \\ 
\midrule
Boneh \textit{et al.} \cite{boneh2019batching} & \multirow{4}{*}{Stateless} & \cmark & \xmark & \xmark & \cmark & \Circle \\ 
MiniChain \cite{chen2020minichain} & & \cmark & \xmark & \xmark & \cmark & \Circle \\ 
SlimChain \cite{xu2021slimchain} & & \cmark & \cmark & \xmark & \cmark & \Circle \\ 
Verkle Tree \cite{verkletree} & & \cmark & \cmark & \xmark & \cmark & \Circle \\
\midrule
BFT-Store \cite{qi2020reliable} & \multirow{3}{*}{EC} & \xmark & \xmark & \cmark & \xmark & \CIRCLE \\ 
PartitionChain \cite{du2023partitionchain} & & \xmark & \xmark & \cmark & \xmark & \CIRCLE \\
Li \textit{et al.} \cite{li2021lightweight} & & \xmark & \xmark & \cmark & \xmark & \CIRCLE \\
\midrule
\textbf{EC-Chain} & State Expiry + EC & \cmark & \cmark & \cmark & \cmark & \CIRCLE \\
\bottomrule[1pt]
\end{tabular}
\begin{tablenotes}
    \item[$\dag$] UTXO or state trie
\end{tablenotes}
\end{threeparttable}
\end{table*}

\section{Related Work}
\label{sec:related:work}
There are four types of approaches to save blockchain storage: pruning, compression, stateless, and erasure coding.

\subsection{Pruning and Compression}
Pruning is a straightforward approach to reduce storage by removing unnecessary data. Both ledger and state data can be pruned. The earliest pruning strategy for blockchain is Simplified Payment Verification (SPV) \cite{nakamoto2008bitcoin}, which requires clients to store only block headers for the fast verification of new blocks. However, its reliance on full nodes and limited accessibility to transaction data can lead to vulnerabilities \cite{faridi2020improving}. For state data, recent research \cite{sforzin2022storage} demonstrates the feasibility of pruning UTXOs in Bitcoin based on the likelihood of them being spent again.
Several compression techniques exist to tackle blockchain storage inefficiency. SCC \cite{kim2019scc} reduces ledger size by merging blocks and is adopted by resource-limited IoT devices. Pontiveros et al. \cite{pontiveros2018recycling} target on comprising smart contract code. Sforzin et al. \cite{sforzin2022storage} propose the MINIMIZE method for Bitcoin, focusing on storing only essential data for unspent transactions, saving storage without sacrificing functionality.

\subsection{Stateless Blockchain and Erasure Coding}
Stateless blockchains aim to maintain a succinct and verifiable proof of states for verification without complicated state maintenance, leveraging techniques Vector Commitment (VC) \cite{boneh2019batching, srinivasan2021hyperproofs, chen2020minichain, xu2021slimchain} and Merkle trees \cite{dryja2019utreexo, bailey2021merkle, chepurnoy2018edrax, verkletree}. Boneh et al. \cite{boneh2019batching} proposed a distributed accumulator with batching for short UTXO commitments, while Hyperproofs \cite{srinivasan2021hyperproofs} is the first maintainable and aggregatable VC scheme. MiniChain \cite{chen2020minichain} introduces the STXO model to reduce storage costs and proving time, and SlimChain \cite{xu2021slimchain} saves on-chain storage by keeping transaction hashes and a Merkle state trie root in blocks. Utreexo \cite{dryja2019utreexo} is a hash-based dynamic accumulator that arranges a UTXO set into a binary Merkle forest, with bridge nodes storing the entire Merkle forest and providing proofs to compact state nodes, which only store the tree roots to save storage. Bailey and Sankagiri \cite{bailey2021merkle} co-locate recent UTXOs in the tree, reducing proof size compared to Utreexo. EDRAX \cite{chepurnoy2018edrax} accelerates stateless blockchains by implementing sparse Merkle trees in the UTXO model and a distributed vector commitment in the account model. The Verkle tree \cite{kuszmaul2019verkle}, combining VC and Merkle trees, is included in Ethereum's roadmap to enhance scalability and sustainability.

Erasure coding, traditionally used for fault-tolerant storage, has gained traction in blockchains as a storage-saving technique. Qi et al. \cite{qi2020reliable} reduce per-block storage overhead from $O(N)$ (proportional to the number of nodes) to a constant $O(1)$, while Du et al. \cite{du2023partitionchain} introduce PartitionChain to minimize the computational cost of encoding and decoding. Additionally, Li et al. \cite{li2021lightweight} address efficient block recovery for resource-constrained nodes, showcasing erasure coding's potential for scalable and efficient blockchain storage.

\subsection{Limitations of Current Approaches}
As illustrated by TABLE \ref{table:comparision}, the current approaches have contributed to storage reduction, but still face sevel limitations:
1) Pruning and Compression: Pruning involves data removal, compromising data availability. Compression techniques though can be lossless, yield only marginal performance enhancements. For example, SLACK \cite{sforzin2022storage} surpasses standard compression methods (gzip, zstd, lzma) in Bitcoin storage, but the maximum achievable savings are limited to 28.62\%.
2) Stateless blockchain: While this method allows validators to store only a constant-size state, it comes with a trade-off. Stateless blockchains don't eliminate the need for state data entirely. Instead, they shift the burden from validators to users. This method requires users to store state data off-chain and frequently update their witnesses, which can negatively impact the availability of on-chain data and cause frequent user interactions\footnote{The number of users affected scales linearly with the total number of transactions on the network.}. More importantly, designing a stateless system can result in significant changes to the current state management system, making implementation challenging.
3) Erasure Coding: This technique reduces data redundancy. However, existing implementations are limited to handling ledger data and do not apply to state data, thus not supporting smart contract-enabled blockchains. Moreover, current approaches lack the capability of adapting to the typical dynamic conditions of permissionless blockchains.

\section{Preliminary and Network Setting}

\subsection{Preliminary}
Here, we present key concepts used in EC-Chain, including blockchain storage, erasure coding, and distributed hash table.

\subsubsection{Blockchain Storage}
A blockchain system mainly contains two types of data: ledger data and state data. In essence, the ledger data provides a complete and tamper-proof record of all blocks, while the state data offers a more efficient way to access and manage the system state based on ledger history. Ledger data are typically implemented using a chain of blocks, where each block contains a set of transactions and a reference (hash) to the previous block. 
Unlike the ledger, the state data is typically implemented with a state trie and maintains the state of every account. We can think of it as an index for the ledger, allowing for fast lookups of the current system state. The state trie is a fundamental data structure underpinning many practical blockchains, with Ethereum being a prime example. Ethereum employs a Merkle Patricia Trie, a special kind of tree structure that leverages hashes for efficient data storage and verification.
%

\subsubsection{Erasure Coding (EC)}
In EC-Chain, we leverage the Reed-Solomon (RS) coding scheme, first described by Reed and Solomon in 1960 \cite{rscode}. This scheme works by manipulating data in units called chunks, often several megabytes (MB) in size. An $(k,m)$-$\mathsf{RS}$ erasure coding scheme allows us to recover original data as long as we have any $k$ chunks from a set of $k + m$ chunks. Specifically, a $(k,m)$-$\mathsf{RS}$ scheme encodes $k$ equal-sized data chunks (denoted as $d_1, \cdots, d_k$) to generate $m$ redundant parity chunks (denoted as $p_1, \cdots, p_m$). These $k$ original data chunks and the $m$ parity chunks together form a \textit{strip}. For convenience, we designate the encoding and decoding functions as $\mathsf{RS.Encode}()$ and $\mathsf{RS.Decode}()$, respectively.


\subsubsection{Distributed Hash Table (DHT)}
A DHT acts as a useful tool for efficiently distributing and retrieving data across a distributed network. Within a DHT, the function $\mathsf{GetClosestPeers}(h)$ identifies the precise node IP addresses that store a data item associated with a specific hash $h$, and $\mathsf{Get}(h)$ retrieves the data fingerprinted by $h$. The Kademlia protocol \cite{maymounkov2002kademlia}, introduced in 2002, is a well-regarded DHT protocol recognized for its efficient routing algorithm and resilience in extensive networks. Therefore, this study utilizes the production implementation of Kademlia provided by Protocol Labs.

\subsection{Network Setting}
In EC-Chain, we consider a network consisting of blockchain storage nodes, referred to simply as nodes within this paper. The network size, represented by $N$, is unkown in permissionless blockchains. The nodes self-organize into groups denoted by $G=\{g_1, g_2, \cdots, g_{\mathcal{K}}\}$. Each group $g_i$ has a size $|g_i|$, and the total network size is given by $N=\sum_{i=1}^{\mathcal{K}} |g_i|$. Each group $g$ is sized $2^t$ for $t \geq 2$. Therefore we also use $g^{(2^t)}$ to denote a group of size $2^t$, with the smallest group comprising four nodes. We set $k = m = |g|/2$. Specifically, for a group of size $2^t$, a $(2^{t-1}, 2^{t-1})$-$\mathsf{RS}$ scheme is employed, enabling data recovery as long as half of the chunks are available. In the subsequent section, we will demonstrate that this configuration is conducive to storage optimization and dynamic network maintenance. 








\section{EC-Chain Design}
\label{sec:design}
EC-Chain reduces storage overheads in permissionless blockchains through a multifaceted approach. Using erasure coding, EC-Chain encodes and distributes ledger and state data between collaborating node groups, significantly reducing individual node storage requirements. Data verifiability is ensured by employing a Distributed Hash Table (DHT), as detailed in Section~\ref{ss:dht_vs_ts}.  For ledger data, optimized encoding strategies like batch encoding and height-based encoding are utilized (Section~\ref{ss:ecledger}). State data benefits from a dual-trie state management system, described in Section~\ref{ss:ectrie}. The fully replicated hot trie facilitates rapid retrieval of frequently accessed data. Less frequently accessed data resides in the cold trie, where erasure coding minimizes storage space. This dynamic system adjusts data placement within the tries based on access patterns, optimizing storage utilization. Additionally, a network maintenance scheme manages node arrivals and departures and incentivizes nodes to merge into larger groups, as Section~\ref{ss:maintenance} illustrates. Overall, this EC-Chain design ensures low storage overheads and high data availability in dynamic permissionless networks.

\subsection{Distributed Hash Table for Data Verifiability}
\label{ss:dht_vs_ts}
Erasure coding aids in decreasing storage costs in blockchain databases by allowing nodes to encode and distribute data fragments in a fault-tolerant manner. To ensure the public verifiability of distributed data in untrusted environments, metadata concerning the data and parity chunks can be managed using either a distributed hash table (DHT) or a threshold signature (TS). Our subsequent analysis demonstrates that, in permissionless blockchain environments (such as EC-Chain), DHTs offer more benefits over TS schemes for two main reasons:

First, establishing a DHT is simpler than setting up a TS. The initialization process involves assigning unique identifiers to nodes, establishing connections to bootstrap nodes, and updating routing tables to optimize data distribution. This approach scales effectively. Establishing a TS group requires a more intricate and time-consuming process compared to DHTs. It involves generating and securely distributing private key shares among known participations. Furthermore, TS schemes exhibit scalability limitations, struggling to efficiently handle large-scale node networks. It takes about 150s to establish 64 nodes for pairing-based threshold cryptosystems \cite{9833584}. TS schemes are inherently more suitable for permissioned blockchains characterized by a known and small set of nodes. 

Secondly, permissionless blockchains exhibit dynamic characteristics, with nodes frequently arriving and departing the network. DHTs excel in such scenarios due to their simple update mechanisms. New nodes seamlessly join by connecting to existing members, updating their routing tables, and redistributing key-value pairs based on their unique identifiers. This ensures scalable performance when distributing data and routing throughout the network. Departing nodes can either notify their neighbors or be identified through health checks to maintain data availability. In contrast, node arrivals or departures within a TS scheme require a complex proactive secret sharing scheme \cite{maram2019churp, baron2015communication}. This entails recalculating threshold parameters, securely distributing new key shares, and updating the public key. This process also requires consensus among remaining nodes to uphold security. 

DHT and TS schemes both require storing proofs of chunks for verification. A notable concern regarding DHT might be the storage overhead. Nevertheless, our estimation suggests that the storage requirement for DHT, even when applied to a large-scale system like Ethereum (approximately 20 million blocks), is estimated to be around 3 GB. This represents only 0.3\% of Ethereum's total storage capacity, indicating that the overhead is relatively minimal.

\subsection{Ledger Encoding}
\label{ss:ecledger}
This paper proposes two advanced encoding strategies: batch encoding and height-based encoding, to optimize processing efficiency and data availability, respectively.

\begin{figure}[!htbp]
    \centering
    \includegraphics[width=0.8\linewidth]{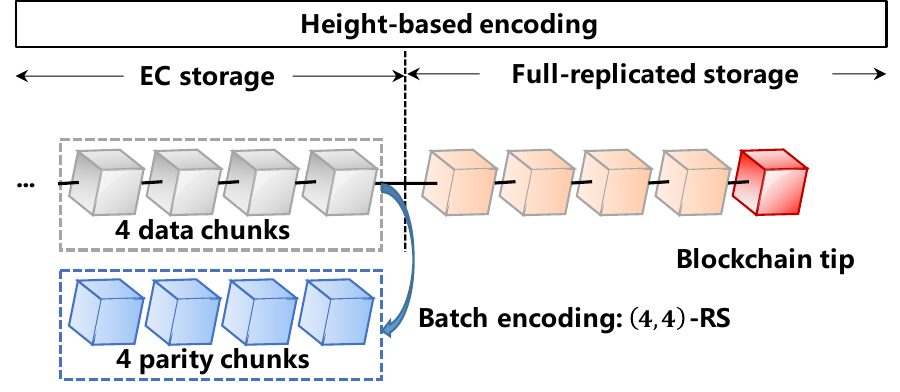}
    \caption{Ledger encoding: an example of applying batch encoding and height-based encoding using $(4,4)$-$\mathsf{RS}$.}
    \label{fig:ledgerencoding}
\end{figure} 

\subsubsection{Batch Encoding}
We can adopt a baseline approach that encodes one block at a time. Specifically, consider $k = m = |g|/2$ and let nodes encode one block. After encoding, a node $p_i$ has a strip of $(k+m)$ chunks denoted $c_1, \cdots, c_{|g|}$, then $p_i$ only keeps the chunk $c_i$ and abandons the other chunks in the same strip. This is a basic instantiation of existing ledger encoding methods \cite{qi2020reliable, du2023partitionchain, li2021lightweight}. While the baseline approach processes blocks correctly, this one-block-at-a-time approach becomes a bottleneck when handling a backlog of blocks requiring encoding. 
To address this limitation, we propose batch encoding as a more effective solution. At the onset of a new encoding process, a leader is selected through the verifiable random function (VRF)~\cite{micali1999verifiable}. This leader is responsible for specifying the height of all blocks requiring encoding. The group encodes $k$ consecutive blocks as a single batch and encodes from the genesis block up to the specified height. For instance, with a group $g^{(4)}(k=2, m=2)$, if the leader designates blocks (1-4) for encoding, the group encodes two separate batches: blocks (1-2) and blocks (3-4). We select consecutive blocks for encoding to facilitate convenient block verification since encoding non-consecutive blocks would result in additional overhead, as missing intermediate blocks would lead to extra batch recovery efforts. Another benefit of batch encoding is that it enables parallel encoding across these multiple batches, thereby enhancing the encoding speed.

\subsubsection{Height-based Encoding} Encoding all blocks without careful consideration can unintentionally undermine data availability. This issue arises because blocks situated near the blockchain tip are prone to be accessed, leading to increased read latency when recovering these blocks frequently. To address this problem, we propose the height-based encoding strategy for the ledger encoding. The basis for this strategy is the observation that typical blockchain access patterns—such as block synchronization \cite{hu2020sync}, transaction validation \cite{chepurnoy2018edrax}, and blockchain exploration \cite{kuzuno2017blockchain}—primarily concentrate on blocks near the current blockchain tip. Therefore, our strategy employs erasure coding for historical blocks that are considerably older than the current blockchain tip. For example, erasure coding can be applied exclusively to blocks that are more than a predefined 10,000 blocks behind the tip. This ensures that only a small fraction of the ledger (e.g., approximately 0.1\%) requires full replication, thereby facilitating the efficient processing of the majority of block requests without incurring large retrieval costs.

\subsection{State Encoding}
\label{ss:ectrie}
In addition to ledger data, state data contains crucial information including account information, contract code, and consensus data for secure and decentralized transaction processing. However, the size of state data grows qucikly, reaching up to 400 GB in the case of Ethereum. To address this issue, we propose a dual-trie state management system, which modifies the existing state trie structure to reduce storage overheads and preserve high data availability.

\begin{figure}[!th]
\centering
\includegraphics[width=0.8\linewidth]{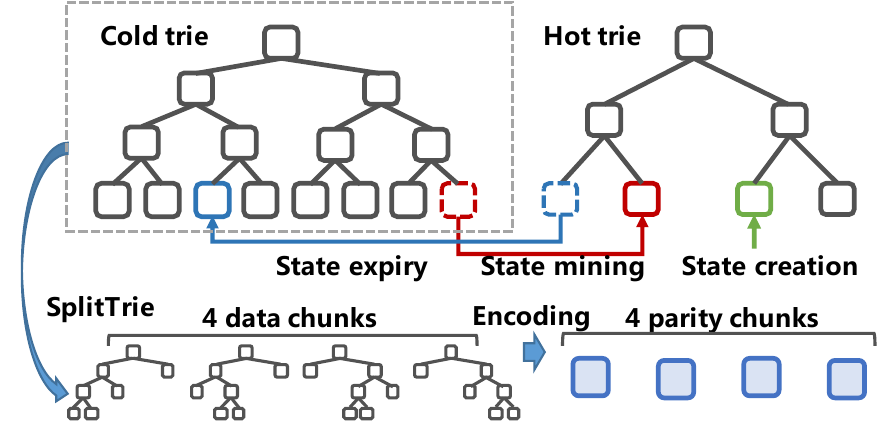}
\caption{The dual-trie state management system and cold trie encoding.}
\label{fig:trie}
\end{figure}

\subsubsection{Dual-Trie State Management} 
We propose an innovative dual-trie state management system that reconciles the trade-offs between storage overheads and data availability inherent in applying erasure coding. This system divides a state trie into two tries: a \textit{hot trie} and a \textit{cold trie}. The hot trie is designed for maintaining frequently accessed states and is fully replicated to prioritize retrieval speed. In contrast, the cold trie employs erasure coding to fragment and distribute less frequently accessed states, thereby optimizing storage utilization. Effective management of these two tries is crucial for maintaining system safety. Hence we design three elaborate procedures: 1) State expiry, which governs the transition of inactive states from the hot trie to the cold trie. 2) State mining, which manages the transfer of states from the cold trie to the hot trie when states become active. 3) State creation, which deals with the establishment of new accounts.

\textit{State Expiry.} 
EC-Chain carefully differentiates between hot and cold states because encoding all states is inefficient and impractical. In the blockchain, blocks are sequentially ordered, typically using a timestamp server approach~\cite{nakamoto2008bitcoin}. This conventional method facilitates a basic categorization of hot and cold states according to the block heights at which the states are generated. However, this method lacks precision, as it relies solely on timestamps and overlooks pertinent blockchain data. For example, an older block might include a frequently accessed state containing a smart contract for a popular NFT application.

\begin{algorithm}[!t]
\caption{The Dual-Trie State Management}\label{alg:2trie}
\DontPrintSemicolon
\SetKw{th}{$\mathcal{T}_\mathsf{hot}$}
\SetKw{tc}{$\mathcal{T}_\mathsf{cold}$}
\SetKw{addr}{$\mathsf{addr}$}
\SetKw{lastaccess}{$\mathsf{lastAccess}$}
\SetKw{accesstime}{$\mathsf{accessTime}$}
\SetKw{expiryheight}{$\mathsf{expiryHeight}$}
\SetKw{addrtoexpire}{$\mathsf{addrToExpire}$}
\SetKw{creationheight}{$\mathsf{creationHeight}$}
\SetKw{timer}{$\mathsf{timer}$}
\KwIn{Hot trie \th, cold trie \tc, new block $b_h$ at height $h$, recency threshold $\Delta T$, frequency threshold $F$, and three maps \accesstime, \creationheight and \timer}
\KwOut{Updated \th and \tc}
\BlankLine
\For{each address \addr to be accessed in $b_h$}{
\If{$\addr \notin \th$}{
\If{$\addr \in \tc$}{
\textcolor{BlueViolet}{// State Mining}\\
Move \addr from \tc to \th\\
}\Else{
\textcolor{BlueViolet}{// State Creation}\\
Add \addr to \th\\
$\creationheight[\addr] \leftarrow h$\\
}
}
Access \addr and $\accesstime[\addr]$++\\
$\expiryheight \leftarrow \max(h + \Delta T, \creationheight[\addr]+\lceil \accesstime[\addr]/F \rceil)$\\
$\timer[\addr] \leftarrow \expiryheight$\\
}
\For{each $\addr$ that $\timer[\addr] = h$}{
\textcolor{BlueViolet}{// State Expiry}\;
Move \addr from \th to \tc\;
}
\end{algorithm}

To address this, we consider two key factors: frequency and recency, for identifying cold/hot accounts. Frequency pertains to the rate at which a state is accessed over a defined period measured by the block height. This method is based on the observation that states with infrequent access are unlikely to be accessed again. However, newly generated states with low access frequency are still likely to experience access in the future. Therefore, recency, which is the distance in height from the last accessed block to the blockchain tip, also becomes a critical factor. By combining these two factors, the system can accurately distinguish cold and hot states. This approach ensures that the majority of state accesses are handled by querying the hot trie, thereby minimizing the need to access the cold trie. A trade-off exists between storage costs and efficiency when determining the thresholds of frequency and recency: $F$ and $\Delta T$. Increasing $F$ and decreasing $\Delta T$ results in more accounts being transferred to the cold trie for erasure coding, which reduces storage costs but increases transaction execution latency due to more frequent cold trie access. The method of organizing state trie is a deterministic algorithm so honest nodes can maintain a consistent view of dual-trie states.

\textit{State Mining and State Creation.} States within the cold trie are less frequently accessed directly, though the possibility of such access cannot be ignored. If a state from the cold trie is accessed, the state is moved to the hot trie as its recency falls below $\Delta T$, and then the proof of state is generated from the hot trie. This procedure, referred to as state mining, lets nodes transit the active states from the cold trie to the hot trie. The creation of a new state needs a blockchain consensus process, which requires Merkle proofs from both tries to verify that the state did not previously exist. This verification step is important to prevent account collisions. Once validated, the new state is inserted into the hot trie, and its frequency and recency are tracked from the block height at which it is created.

\begin{algorithm}[!t]
  \DontPrintSemicolon
  \SetKwFunction{splittrie}{$\mathsf{SplitTrie}$}
  \SetKw{st}{$\mathsf{st}$}
  \SetKw{rt}{$\mathsf{rt}$}
  \SetKw{mp}{$\mathsf{mp}$}
  \SetKw{root}{$\mathsf{root}$}
  \SetKw{subtrie}{$\mathsf{subtrie}$}
  \SetKw{subtrieroots}{$\mathsf{subtrieRoots}$}
  \SetKw{merklepath}{$\mathsf{merklePath}$}
  
  \textcolor{BlueViolet}{// \splittrie}\\
  \KwIn{State trie $\mathcal{T}$ with \rt as the root, and an integer $k$}
  \KwOut{Data chunks $\{d_1, d_2, \cdots, d_k\}$}
  \subtrieroots = $\{\rt\}$\\
  \While{$|\subtrieroots| < k$}{
  \For{each $\root \in \subtrieroots$}{
  Remove \root from \subtrieroots\\
  Add all children of \root to \subtrieroots\\
  }
  }
  \For{each $\root_i \in \subtrieroots$}{
  $\subtrie_i \leftarrow$ subtrie of $\mathcal{T}$ rooted at $\root_i$\\
  $\merklepath_i \leftarrow$ Merkle path between \rt and $\root_i$\\
  $d_i \leftarrow \subtrie_i \| \merklepath_i$\\
  }
  
  \textcolor{BlueViolet}{// Encoding}\\
  \KwIn{State trie $\mathcal{T}$, and group $g_i$}
  $k \leftarrow m \leftarrow |g_i|/2$\\
  $\mathcal{ST} \leftarrow \{\st_1, \st_2, \cdots, \st_k\} \leftarrow \splittrie(\mathcal{T}, k)$\\
  $\mathcal{C} \leftarrow \{c_1, c_2, \cdots, c_{k+m}\} \leftarrow \mathsf{RS.Encode}(\mathcal{ST}, k, m)$\\
  \For{$c_j \in \mathcal{C}$}{
    $h_j \leftarrow H(c_j)$\\
    \If{$i \in \mathsf{DHT}.\mathsf{GetClosestPeers}(h_j)$}{
    store $c_j$\\
    }
    $\mathcal{H} \leftarrow \mathcal{H} \cup \{h_j\}$\\
  }
  Store $\mathcal{H}$\\
  
  \textcolor{BlueViolet}{// Decoding}\\
  \KwIn{hashes of all chunks $\mathcal{H}$}
  \For{$h_j \in \mathcal{H}$}{
    $c_j \leftarrow \mathsf{DHT}.\mathsf{Get}(h_j)$, $\mathcal{C} \leftarrow \mathcal{C} \cup c_j$\\
    \If{$|\mathcal{C}| \ge k$}{
    $T \leftarrow \mathsf{RS.Decode}(\mathcal{C})$\\
    }
  }
  \caption{Cold Trie Encoding}
  \label{alg:ectrie}
\end{algorithm}

\textit{State Transitions in the Dual-Trie System.} 
For ease of understanding, we outline the complete state transition process in Algorithm~\ref{alg:2trie}. During transaction processing, the system verifies a state by querying the hot trie, the cold trie, or treating it as non-existent. When a state from the cold trie is accessed, it is transferred to the hot trie through state mining to support faster future retrieval (lines 3-5). For a new account, the system creates and inserts the state directly into the hot trie (lines 7-9). Any state access results in an update to the state's metadata for expiry tracking. The expiry estimation uses a block height $h$ at which the state is accessed and a predefined recency threshold $\Delta T$. Considering recency, the system determines that the state should expire at $h + \Delta T$, given no intervening accesses. Considering access frequency, the system maintains two values: $\mathsf{creationHeight}[\mathsf{addr}]$, which records the block height of the state's creation, and $\mathsf{accessTime}[\mathsf{addr}]$, which accumulates the total access count. Using these values and the access frequency threshold $F$, the system calculates the expected block height when the access frequency falls below the threshold: $\mathsf{creationHeight}[\mathsf{addr}]+\lceil\mathsf{accessTime}[\mathsf{addr}]/F\rceil$. This calculated value, along with the recency consideration, determines the final expiry timer for the address (lines 10-12). After processing all transactions within a block, the system iterates through states and expires those whose expiry timers coincide with the current block height (lines 13-15).

\subsubsection{Cold Trie Encoding} 


As a subsequent step, the cold trie can be erasure coded to save storage as Algorithm~\ref{alg:ectrie} depicts. If the cold trie is directly serialized into a byte array and then divided into chunks, responding to each query requires the recovery of the entire cold trie, which is time-consuming. We observe that tries employ constant-sized hashes and are typically balanced. For instance, the sizes of 16 subtries, each rooted at one of the sixteen children of the MPT root are almost equal as illustrated in Fig.~\ref{fig:balance}.
To achieve data availability, we design the function $\mathsf{SplitTrie}$ (lines 1-10) to partition the cold trie into multiple balanced subtries. These subtries serve as data chunks and are used as inputs into the encoding function. Given that $k=|g|/2=2^{t-1}$ and the cold trie is balanced, the cold trie can be segmented into $k$ subtries of approximately equal size. This process ensures that each data chunk remains operational despite that the entire cold trie is encoded. To enable a node storing a data chunk to verify its origin from the complete cold trie, the Merkle path from the subtrie's root to the root of the cold trie is included as proof with the data chunk (lines 9-10). Finally, a $(k,m)$-$\mathsf{RS}$ code is applied to generate $m$ parity chunks, with the total $k+m$ chunks as a strip (lines 11-20).

\begin{figure}[!htbp]
    \centering
    \includegraphics[width=0.8\linewidth]{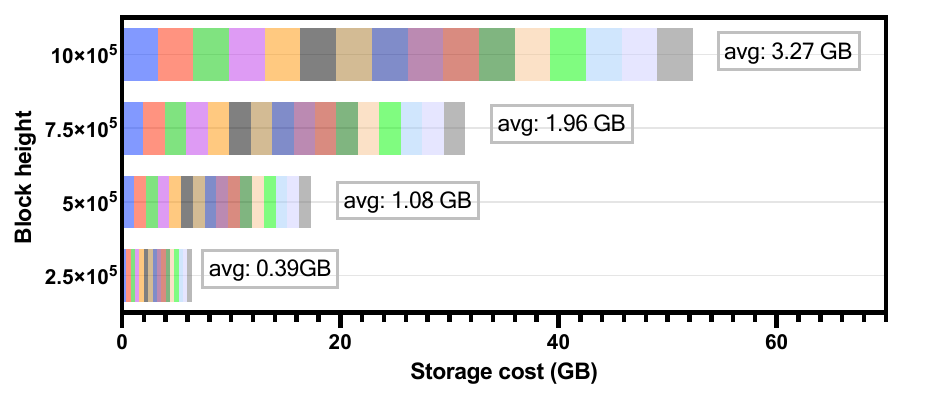}
    \caption{The sizes of 16 subtries (represented by 16 distinct colors) of the Merkle Patricia Trie (MPT) within the initial one million blocks of Ethereum.}
    \label{fig:balance}
\end{figure}

\subsection{Network Maintenance for Dynamic Networks}
\label{ss:maintenance}
The dynamic nature of node participation in permissionless networks presents a significant challenge in reducing storage. Frequent node arrivals and departures require computationally intensive EC re-encoding, hindering network efficiency and data availability. To mitigate this, we propose a novel network maintenance scheme. As detailed in Sections~\ref{ss:ecledger} and \ref{ss:ectrie}, larger EC groups require less storage space at each node. This inherent characteristic incentivizes nodes to join larger groups, thereby minimizing their individual storage consumption. However, larger groups increase the risk of unreliable storage. Therefore, we design a new simple rule: larger groups must be formed by merging two smaller ones. This rule ensures that a node must remain within the network for a prolonged period to become part of a larger group, effectively mitigating the risk associated with short-lived participants.

\textit{Group Upgrade and Downgrade.}
The network maintenance protocol operates as follows: When two groups of identical size are identified, they can merge into a larger single group, enjoying the lower redundancy. Based on the network setting, the number of groups is directly proportional to the number of `1' bits in the binary representation of $N$. For example, a network with 136 nodes would finally be organized into two stable groups: one large group $g^{(128)}$ (encoded using $(64, 64)$-$\mathsf{RS}$) and another smaller group $g^{(8)}$ (encoded using $(4, 4)$-$\mathsf{RS}$). This arrangement reflects the binary representation of 136, which is $10001000$. Given that the storage costs for maintaining the blockchain in a group scale with the group's size, the overall storage costs for the network show a linear relationship with the number of groups. This number can range from an optimal scenario of 1 to a worst-case scenario of $\log(N)$, with an average of $\log(N)/2$. The analysis is elaborated in Section~\ref{sec:analysis}.

If a node leaves the network or encounters a malfunction, two scenarios arise: temporary absence and prolonged absence. A node may be temporarily offline with the expectation of returning soon. During this period, the group can continue its operations without disruption, as the redundancy provided by the EC scheme ensures continued access to blockchain data. However, if the number of active nodes drops below a critical threshold, the network would be vulnerable to Byzantine attacks. Therefore, when the rate of node departure in a group exceeds 1/4 of its total members, a downgrade process is initiated to maintain the group's resilience against up to 1/3 of nodes potentially exhibiting Byzantine behavior. To downgrade a group that has lost a quarter of its nodes, the remaining 3/4 of members are reorganized into two smaller groups. One group will consist of 1/4 of the remaining nodes, and the other will include 1/2. This downgrade ensures that both new groups adhere to a size $2^t$, facilitating smooth future adjustments.

\textit{Chunk Update.} 
Few investigations have addressed the dynamics of chunk updates during data loss and recovery. When a group upgrade or downgrade occurs, the blockchain data within these groups must be modified to conform to the new configuration. Specifically, the transition from a $(k, m)$-$\mathsf{RS}$ scheme to a $(2k, 2m)$-$\mathsf{RS}$ scheme requires transforming the stored chunks, involving merging adjacent $k$-block batches into $2k$-block batches. Consider two groups, $g_1$ and $g_2$, using a $(k, m)$-$\mathsf{RS}$ scheme, which merge to form a new group $g_3$ under a $(2k, 2m)$-$\mathsf{RS}$ scheme. The nodes $g_1^k$ and $g_2^k$ hold the original data chunks, while the nodes $g_1^m$ and $g_2^m$ store the corresponding parity chunks. In EC-Chain, data blocks are organized into batches denoted as $(B_1, B_2, \cdots)$, with each batch consisting of $2k$ blocks. For data chunk update, $g_1^k$ retains only odd-numbered batches, while $g_2^k$ retains only even-numbered batches. Consequently, upon merging into $g_3$, $g_1^k$ and $g_2^k$ collectively possess the encoded $2k$-block batches, obviating the need for data chunk redistribution. Subsequently, nodes $g_1^m$ and $g_2^m$ can request the necessary data chunks from $g_1^k$ and $g_2^k$ to generate the parity chunks using $(2k, 2m)$-$\mathsf{RS}$.

The transition from a $(k, m)$-$\mathsf{RS}$ scheme to a $(2k, 2m)$-$\mathsf{RS}$ scheme for state data requires a split operation on the subtries. Since $g_1^k$ and $g_2^k$ share the same encoded cold trie, we can consider the node pairs $(p_1, p_2)$ where $p_1 \in g_1^k$ and $p_2 \in g_2^k$ handle the same subtrie. Each node within such a pair independently divides its subtrie into two balanced subtries, denoted as $\mathcal{T}_1$ and $\mathcal{T}_2$. Thereafter, $p_1$ retains $\mathcal{T}_1$ and $p_2$ retains $\mathcal{T}_2$. Finally, both $g_1^k$ and $g_2^k$ acquire the encoded state data according to the new $(2k, 2m)$-$\mathsf{RS}$ scheme and proceed to generate parity chunks from these data chunks. The process of updating the ledger and state data during a downgrade is the reverse of the upgrade process, thus we do not repeat it for brevity.

\section{Analysis on Redundancy and Bandwidth Cost}
\label{sec:analysis}

We analyze the storage redundancy and bandwidth consumption of EC-Chain from a theoretical perspective.

\subsection{Storage Redundancy}
We provide a formal definition of redundancy within a blockchain network.

\begin{definition}[Redundancy]
Let $S$ represent the storage requirement for a node to maintain the ledger and state data. The redundancy of a blockchain network is defined as the total storage costs in all nodes, denoted $k \cdot S$. For full replication, the redundancy is $N \cdot S$, where $N$ is the number of nodes.
\end{definition}

\begin{theorem}
The redundancy of EC-Chain is upper bounded by $2\lceil \log N \rceil \cdot S$ and lower bounded by $2S$, with an expected value of $\lceil \log N \rceil \cdot S$.
\end{theorem}
\begin{proof}
For any group $g$ of size $2^t$ ($t\geq 2$), a $(2^{t-1}, 2^{t-1})$-$\mathsf{RS}$ code is used to store blockchain data, resulting in a storage overhead of $(2^{t-1} + 2^{t-1}) \cdot \frac{S}{2^{t-1}} = 2S$. This indicates that the storage overhead per group is $2S$, regardless of the group's size. So the redundancy of EC-Chain depends solely on the number of groups. According to the network maintenance protocol, groups of the same size merge into larger groups. Hence, under stable conditions, there will not be multiple groups of identical sizes, with all group sizes being powers of 2. The group structure in EC-Chain can be interpreted as a binary decomposition of the total number of nodes $N$, where the number of groups corresponds to the number of 1s in the binary representation of $N$. For a randomly distributed $N$, the number of groups falls within $[1, \lceil \log N \rceil]$. Thus, the redundancy of EC-Chain is upper bounded by $2\lceil \log N \rceil \cdot S$, lower bounded by $2S$, and has an expected value of $\lceil \log N \rceil \cdot S$. This represents a significant reduction compared to the $N \cdot S$ redundancy associated with full replication.
\end{proof}

\subsection{Bandwidth Cost}
To analyze the impacts of dynamics, we define the node arrival rate as $\alpha$ ($\alpha \geq 4$ w.l.o.g.) and the departure rate as $\beta (\beta\geq 0)$. The relationship between $\alpha$ and $\beta$ delineates two distinct scenarios of network evolution: rapid growth ($\alpha > \beta$) and dynamic equilibrium ($\alpha = \beta$). The rapid growth scenario signifies a significant increase in the number of nodes, typically observed during periods of blockchain expansion. In contrast, the dynamic equilibrium scenario reflects a stable network with a constant user base. These scenarios cover the primary conditions encountered by blockchains.

\begin{theorem}
The total bandwidth cost of EC-Chain is upper bounded by $(\frac{5}{4}\alpha - 1 + \lfloor \log N \rfloor + \beta)S$ and lower bounded by $(\frac{5}{4}\alpha - \lceil \log \frac{\alpha}{4} \rceil)S$.
\end{theorem}

\begin{proof}
When nodes join EC-Chain, they initially require a complete copy of the blockchain data from a group that can fully recover the data, resulting in a bandwidth cost of $\alpha S$. Subsequently, group upgrades might occur. Newly added nodes also merge into new groups until no two groups have the same size when they become stable. This process involves $\lfloor \frac{\alpha}{4 \times 2} \rfloor + \lfloor \frac{\alpha}{8 \times 2} \rfloor + \cdots = \frac{\alpha}{4} - \mathsf{wt}(\frac{\alpha}{4})$ group upgrades, where $\mathsf{wt}(x)$ represents the Hamming weight of an integer $x$, i.e. the number of 1s in the binary representation of $x$, with $1 \le \mathsf{wt}(x) \le \lceil \log x \rceil$. In the worst case, the addition of one node could trigger upgrades in all old groups until they merge into a single group, leading to $\lfloor \log N \rfloor$ group upgrades and a total cost of $\lfloor \log N \rfloor S$. Hence, the worst-case bandwidth cost of adding a node is $(\frac{5}{4}\alpha - 1 + \lfloor \log N \rfloor)S$. In the best case, no old groups are upgraded, resulting in a bandwidth cost of $(\frac{5}{4}\alpha - \lceil \log \frac{\alpha}{4} \rceil)S$. Despite these costs, group upgrades can occur in parallel.
In EC-Chain, a group will be downgraded and require data recovery if it loses 1/4 of its nodes. In the worst-case scenario, $\beta$ node departures could lead to $\beta$ groups being downgraded, resulting in a bandwidth cost of $\beta S$. In the best case, departures do not trigger any downgrade, so there will be no additional cost. Combining the costs associated with the arrivals and departures of the nodes, the overall bandwidth cost is bounded by $(\frac{5}{4}\alpha - 1 + \lfloor \log N \rfloor + \beta)S$ and bounded by $(\frac{5}{4}\alpha - \lceil \log \frac{\alpha}{4} \rceil)S$.
\end{proof}

\section{Evaluation}
\label{sec:exp}

We conduct experiments to evaluate the performance of EC-Chain and compare it with Ethereum under various configurations. 

\subsection{Implementation and Experiment Setup}
We implement EC-Chain on top of go-ethereum\footnote{https://github.com/ethereum/go-ethereum} (commonly known as geth), the most widely used Ethereum implementation developed in Golang. As illustrated in Fig.~\ref{fig:implemetation}, we integrate the ledger encoding module with 631 lines of code (LOCs) and state encoding module with 967 LOCs into Ethereum's blockchain database. Furthermore, we code a network maintenance module of 429 LOCs and integrate it into Ethereum's p2p module responsible for network management. Besides, the RS coding library is provided by Klaus Post\footnote{https://github.com/klauspost/reedsolomon}, and the DHT library is supplied by Protocol Labs\footnote{https://github.com/libp2p/go-libp2p-kad-dht}.
\begin{figure}[!htbp]
    \centering
    \includegraphics[width=0.8\linewidth]{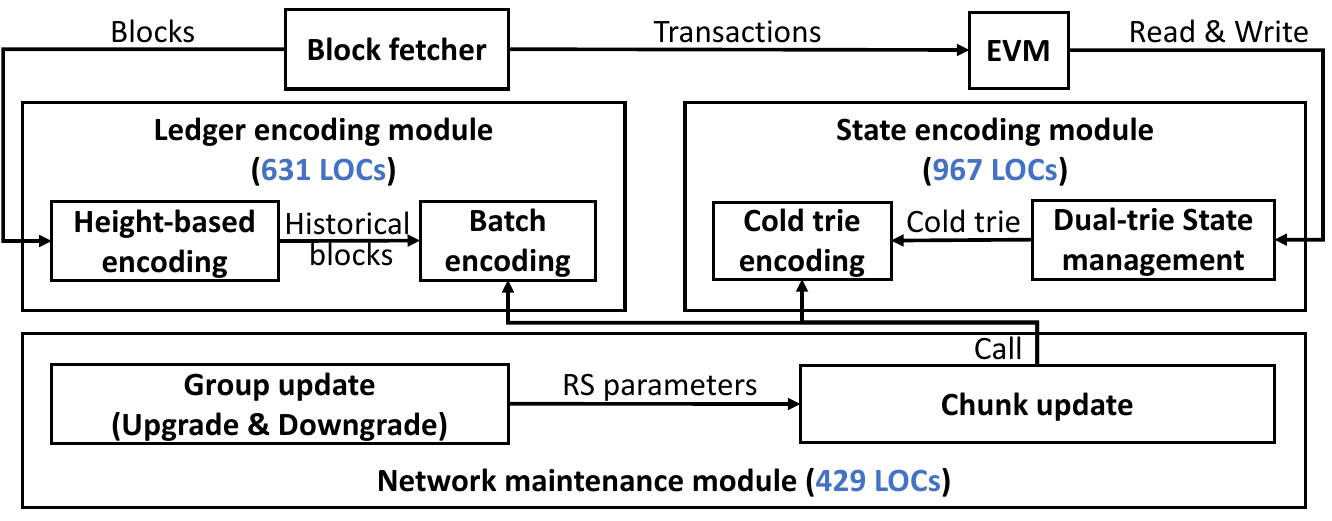}
    \caption{The implementation of EC-Chain}
    \label{fig:implemetation}
\end{figure}

EC-Chain's performance is primarily influenced by: the network size $N$, the recency threshold $\Delta T$, and the frequency threshold $F$. Therefore, our evaluation included testing EC-Chain's performance with varying group sizes and thresholds. To demonstrate the performance improvement brought by EC-Chain, we use go-ethereum as a baseline in our evaluation. We select transactions from the first 4 million blocks of the Ethereum mainnet and replay them during the evaluation. In terms of the running environment, up to 64 blockchain nodes are distributed across 17 regions in 10 countries. Each node is equipped with an 8-Core CPU, 32 GB of memory, and a 512 GB NVMe SSD, all running Ubuntu 22.04 LTS. Each node has a bandwidth of 1 Gbps. Additionally, we run a client in full sync mode to synchronize transactions from the Ethereum mainnet and forward them to the blockchain nodes for transaction replay. 


\subsection{Storage Costs}

First, we evaluate the storage costs of EC-Chain. Fig.~\ref{fig:storage} shows the average storage costs per node in the network after replaying transactions in each block. Fig.~\ref{fig:storage_n} illustrates the impact of $N$ on EC-Chain's storage costs, with thresholds set to $\Delta T=10^4$ and $F=10^{-2}$. The storage per node decreases with the increasing $N$. When $N=64$, at a block height of 4 million, the storage costs per EC-Chain node is 91.8\% lower than that of an Ethereum node. The results also indicate that doubling $N$ nearly halves the storage costs per node. This reduction occurs mainly because data is encoded by $(k, m)$-$\mathsf{RS}$ with $k=m=|g|/2$. 
Fig.~\ref{fig:storage_threshold} depicts the storage costs across various thresholds, with $N$ held constant at 8. The findings reveal that increasing the frequency threshold $F$ and decreasing the recency threshold $\Delta T$ lead to reduced storage costs, as more states are moved to the cold trie for erasure coding. 

\begin{figure}[!t]
\centering
\subfigure[$\Delta T=10^4$, $F=10^{-2}$]{
\label{fig:storage_n}
\includegraphics[width=0.48\linewidth]{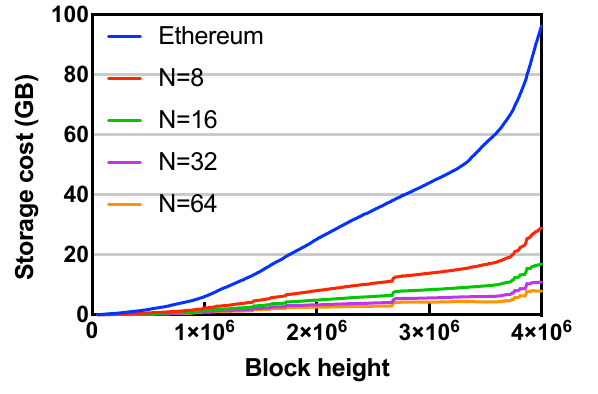}}
\subfigure[$N=8$]{
\label{fig:storage_threshold}
\includegraphics[width=0.48\linewidth]{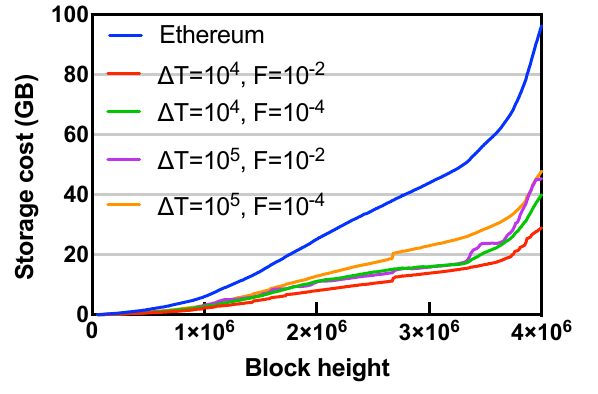}}
\caption{The storage costs per node.}
\label{fig:storage}
\end{figure}

\begin{figure}[!t]
\centering
\subfigure[Group upgrade]{
\label{fig:group_upgrade}
\includegraphics[width=0.48\linewidth]{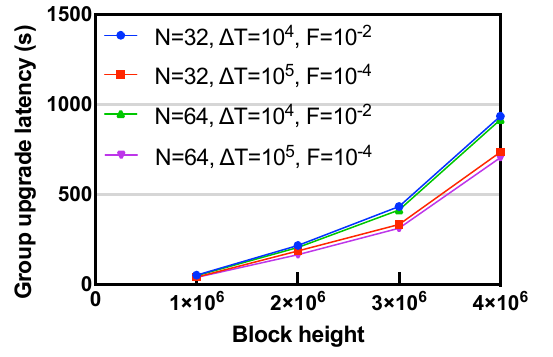}}
\subfigure[Group downgrade]{
\label{fig:group_downgrade}
\includegraphics[width=0.48\linewidth]{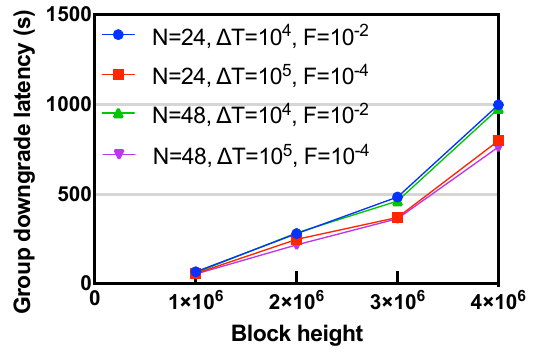}}
\caption{The latency of group upgrade and downgrade.}
\label{fig:maintenance}
\end{figure}

\begin{figure}[!t]
\centering
\subfigure[$\Delta T=10^4$ and $F=10^{-2}$]{
\label{fig:latency_n}
\includegraphics[width=0.46\linewidth]{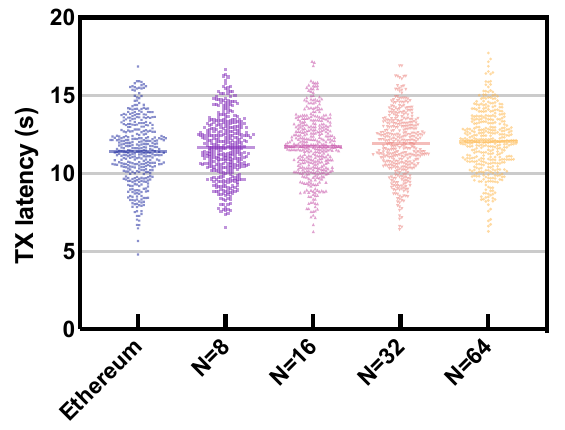}}
\subfigure[$N=8$]{
\label{fig:latency_threshold}
\includegraphics[width=0.46\linewidth]{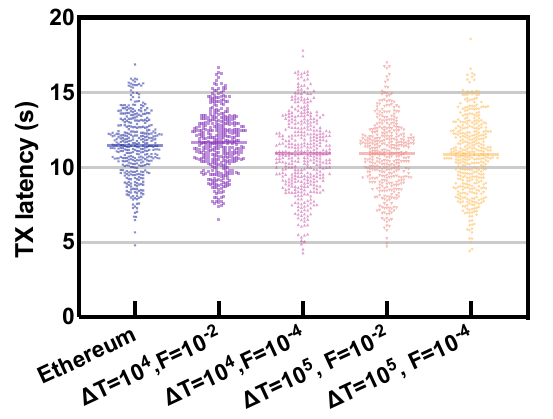}}
\caption{The transaction latency.}
\label{fig:latency}
\end{figure}

\begin{figure}[!t]
\centering
\subfigure[$\Delta T=10^4$ and $F=10^{-2}$]{
\label{fig:tps_n}
\includegraphics[width=0.46\linewidth]{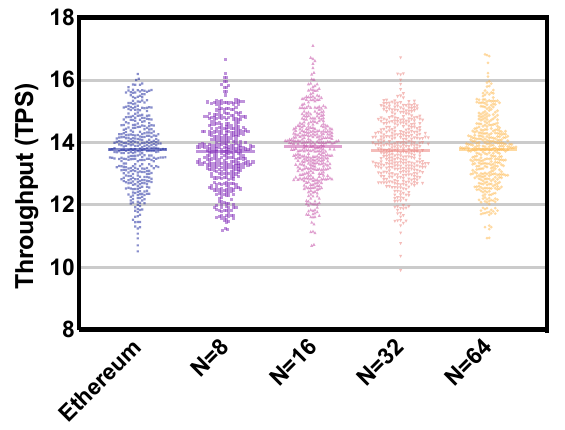}}
\subfigure[$N=8$]{
\label{fig:tps_threshold}
\includegraphics[width=0.46\linewidth]{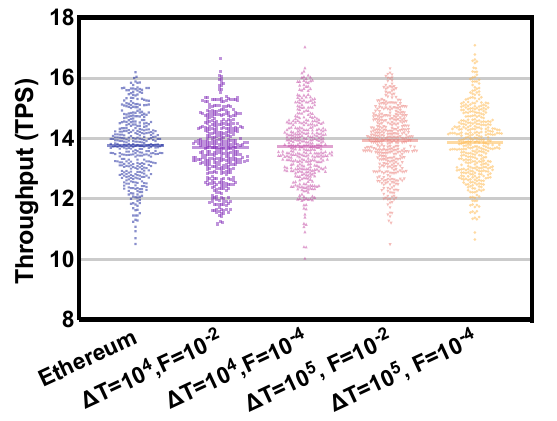}}
\caption{The throughput.}
\label{fig:tps}
\end{figure}

\subsection{Latency of Network Maintenance}

We evaluate the latency of network maintenance as shown in Fig.~\ref{fig:maintenance}. For upgrades, we test the latency of combining two $g^{(16)}$ into one $g^{(32)}$, and two $g^{(32)}$ into one $g^{(64)}$. For downgrades, we examine the performance of downgrading a $g^{(32)}$ (8 nodes leaving, 24 left) into $g^{(16)}$ and $g^{(8)}$, and downgrading $g^{(64)}$ (16 nodes leaving, 48 left) into $g^{(32)}$ and $g^{(16)}$. We also evaluate how varying thresholds impact these processes.
The latency increases as the blockchain expands due to the extended time required for chunk updates caused by the larger volume of ledger and state data. When $\Delta T$ increases and $F$ decreases, a greater number of states remain in the hot trie, leading to reduced latency for both group upgrades and downgrades, as the fully replicated hot trie eliminates the need for updates. Moreover, larger groups exhibit lower latency compared to smaller groups because they benefit from more concurrent data transfers, optimizing the use of available bandwidth.
Additionally, by comparing Fig.~\ref{fig:group_upgrade} and Fig.~\ref{fig:group_downgrade}, it is evident that the latency for upgrades is slightly lower than for downgrades. This arises because, during group upgrades, most data chunks are accessible. In contrast, during downgrades caused by node disparity, there is a higher likelihood of missing data chunks, requiring additional decoding efforts from the remaining parity chunks.

\subsection{Transaction Latency and Throughput}

We measure the transaction latency in second and throughput in transactions per second (TPS), as shown in Fig.~\ref{fig:latency} and Fig.~\ref{fig:tps}. The latency is defined as the time interval from when a client sends a transaction until the transaction is confirmed by the blockchain. 
In Fig.~\ref{fig:latency_n}, the transaction latency between Ethereum and EC-Chain (for varying values of $N$) is analyzed, with EC-Chain's parameters configured to $\Delta T=10^4$ and $F=10^{-2}$. EC-Chain under different $N$ exhibits similar latency to Etheruem, suggesting minimal impact of $N$ on transaction latency.
Figure~\ref{fig:latency_threshold} presents the latency under various thresholds, with $N$ fixed at 8. When EC-Chain's thresholds are set to $\Delta T=10^4$ and $F=10^{-2}$, the latency between EC-Chain and Ethereum is nearly identical. As $F$ decreases or $\Delta T$ increases, a greater number of states are maintained in the hot trie, allowing for faster verification.
Another key performance indicator for a blockchain, throughput, is depicted in Figure \ref{fig:tps}. We also examine this metric under varying network sizes (shown in Figure \ref{fig:tps_n}) and different thresholds (shown in Figure \ref{fig:tps_threshold}). The results demonstrate that in all scenarios, EC-Chain maintains a throughput nearly equal to that of Ethereum, indicating that EC-Chain does not experience a decrease in throughput.




\section{Conclusion}

We propose EC-Chain, a cost-effective storage solution for permissionless blockchain systems that optimizes storage using erasure coding. EC-Chain can minimize the storage costs associated with both ledger data and state data. To reduce the storage costs of state data, EC-Chain introduces a dual-trie state management system, which migrates inactive states to a cold trie, which can be erasure-coded but ensures high data availability. Additionally, we design a network maintenance strategy to achieve adaptability when keeping low data redundancy. 



\bibliographystyle{IEEEtran}
\bibliography{references}

\end{document}